\newif\ifexternalize

\documentclass[a4paper]{article}

\usepackage[utf8]{inputenc} 

\usepackage{a4wide}
\usepackage{xcolor,graphicx}

\usepackage{amsmath}
\usepackage{amsthm}
\usepackage{amsfonts}
\usepackage{mathtools}
\usepackage[backend=bibtex,
	style=numeric, 
	hyperref=true, 
	natbib=true, 
	backrefstyle=three, 
	sortcites=true, 
	maxbibnames=50,
	maxcitenames=3, 
 	firstinits=true, 
	isbn=false,url=false,doi=false] 
	{biblatex} 
	
\bibliography{bib} 
\newbibmacro{string+doiurl}[1]{%
	\iffieldundef{doi}{%
		\iffieldundef{url}{%
			#1%
        }{%
			\href{\thefield{url}}{#1}%
		}%
	}{%
		\href{http://dx.doi.org/\thefield{doi}}{#1}%
	}%
}

\usepackage[switch]{lineno}  
\usepackage{todonotes}

\usepackage[pdfdisplaydoctitle,menucolor=orange!40!black,filecolor=magenta!40!black,urlcolor=blue!40!black,linkcolor=red!40!black,citecolor=green!40!black,colorlinks]{hyperref}

\usepackage[nameinlink,capitalize]{cleveref}

\usepackage{booktabs}
\usepackage{pgfplots}
\usepgfplotslibrary{groupplots}
\usepackage{siunitx}
\usepackage{pgfplotstable}

\ifexternalize
	\usetikzlibrary{external}

	\immediate\write18{mkdir -p tikz-figs/} 
\fi

\newtheorem{theorem}{Theorem}

\newtheorem{Definition}{Definition}

\newcommand{\bigO}{\mathcal{O}} 
\DeclareMathOperator{\dist}{dist}

\newcommand{\Q}{Q}

\begin{document}

\newcommand\relatedversion{}
\renewcommand\relatedversion{\thanks{The full version of the paper can be accessed at \protect\url{https://arxiv.org/abs/TODO}}} 

\title{\Large Correlating Theory and Practice in Finding Clubs and Plexes}
\author{Aleksander Figiel \and Tomohiro Koana \and André Nichterlein \and Niklas Wünsche}
\date{TU Berlin, Faculty IV, Algorithmics and Computational Complexity, Berlin, Germany\\
\texttt{\{a.figiel,tomohiro.koana,andre.nichterlein\}@tu-berlin.de}}

\maketitle







\begin{abstract} \small\baselineskip=9pt 
	Finding large ``cliquish'' subgraphs is a classic NP-hard graph problem.
	In this work, we focus on finding maximum $s$-clubs and $s$-plexes, i.\,e., graphs of diameter~$s$ and graphs where each vertex is adjacent to all but~$s$ vertices.
	Preprocessing based on Turing kernelization is a standard tool to tackle these problems, especially on sparse graphs.
	We provide a new parameterized analysis for the Turing kernelization and demonstrate their usefulness in practice.
	Moreover, we provide evidence that the new theoretical bounds indeed better explain the observed running times than the existing theoretical running time bounds.
	To this end, we suggest a general method to compare how well theoretical running time bounds fit to measured running times.
	
\end{abstract}

\section{Introduction}\label{sec:intro}

Highly engineered solvers perform often much better than the known theoretical results would suggest. 
This is especially true when dealing with NP-hard problems. 
Unless P = NP, no efficient (i.\,e.\ polynomial-time) algorithm exists that solves all input instances correctly.
However, optimized implementations can often solve instances with millions of vertices, variables, etc.\ as demonstrated frequently at algorithm engineering conferences; see for example \citet{BFSW22,ST22} for two examples from last year's ALENEX.
Of course, these implementations are not polynomial-time algorithms for NP-hard problems.
The real-world instances are simply not those that require the worst-case runtime.
On the other hand, there are usually small instances making these solvers struggle. 
So theoretical running time bounds do not match observed running times on the given data set.
Obviously, a better connection between theoretical results and empirical findings would be highly desirable.

A multivariate (i.\,e.\ parameterized) analysis of the algorithm allows for a more nuanced picture of running time bounds. 
In principle, it could provide us with a much better prediction for the running time.
However, a comparison to the theoretical parameterized running time is rarely made in practice (although there are notable exceptions~\cite{WB20,KKNNZ21}).
This is probably due to the multitude of issues arising here; let us mention just a few:
For example, most theoretical bounds are stated using the~$O$-notation that hides constants. 
Matching these to observed running times (which depends also on the used hardware) is not straight forward.
Moreover, there are often several different parameterized algorithms which even could have some overlap in their approach, that is, the observed running time most likely depend on many parameters.
In this work, we propose an approach addressing these issues. 
It allows us to compare (roughly) which theoretical running-time bound fits ``better'' to the observed running times for a given data set.

We exemplify our approach on the $s$-\textsc{Club} and $s$-\textsc{Plex} problems and show how for various solver variants different theoretical explanations can be used. 
To this end, we follow the approach of~\citet{WB20} who demonstrated by means of a multivariate analysis why \textsc{Clique} is often efficiently solvable in relatively sparse graphs.

\subsection{Related work}

\paragraph{\textsc{Clique} on Sparse Graphs.}
\textsc{Clique} is one of Karp's 21 NP-complete problems~\cite{Kar72}.
As such, it is well studied, both in theory and practice; see \citet{WH15} for a survey.
The currently fastest exact algorithm has running time~$O(1.20^n)$ \cite{XN17}, where~$n$ is the number of vertices.
While~$1.20$ seems very small, for a graph with 400 vertices the number of steps has more than 30 \emph{digits} which is still infeasibly large.

It is easy to see that any clique is contained in the neighborhood of each of its vertices.
Thus, a very basic approach solving clique on a sparse graph~$G=(V,E)$ is the following.
Take a vertex~$v$ of minimum degree and find the largest clique in~$N[v]$ (the closed neighborhood of~$v$).
Then, remove~$v$ and continue in the same fashion. 
In the end, output the largest found clique.
The \emph{degeneracy}~$d$ of a graph is the size of the largest neighborhood encountered in the above algorithm.
Hence, the above algorithm can be implemented to run in~$1.20^d \cdot n^{O(1)}$ time which is on large sparse graphs far better than the~$O(1.20^n)$ bound.
Many of the graphs considered by \citet{WB20} have several hundred thousand vertices and can be solved in less than a minute (often less than a second).
Yet, some of these graphs have a degeneracy of well above 400 (again resulting in an infeasibly large number of steps).
To rectify this, \citet{WB20} provide an algorithm running in~$1.28^gn^{O(1)}$ time where~$g:= {d-k+1}$ is called the \emph{core-gap} and~$k$ denotes the number of vertices in a maximum clique (see \cref{sec:theory} for a more detailed explanation).
Clearly~$g$ can be much smaller than~$d$.
In fact, \citet{WB20} observe that all their relatively small but hard-to-solve instances have a large core-gap.

\paragraph{Clubs and Plexes.}
An $s$-club is a graph of diameter~$s$. 
An $s$-plex is a graph with~$\ell$ vertices where every vertex has degree at least~$\ell-s$. 
While not required by definition, in this work we only consider \emph{connected} $s$-plexes. 
The task in $s$-\textsc{Club} / $s$-\textsc{Plex} is to find the largest $s$-club / $s$-plex in a given graph.

Both $s$-\textsc{Club} and $s$-\textsc{Plex} are NP-hard as they contain \textsc{Clique} as special case ($s=1$).
Both problems are well-studied in the literature, both from theoretical and practical perspective.
For example, $s$-\textsc{Plex} is W[1]-hard with respect to the parameter solution size~$k$ for all~$s \ge 1$~\cite{KR02,KHMN09}.
In contrast, if~$s > 1$, then $s$-\textsc{Club} is fixed-parameter tractable~\cite{SKMN12,CHLS13}.
We refer to \cite{Kom16} for a further overview on the parameterized complexity of these problems.
%
Several algorithmic approaches (heuristics and exact algorithms) have been proposed and examined to find maximum-cardinality 2-clubs~\cite{BLP02,BLP00,BS17,CHLS13,HKN15,MB12} or 2-plexes~\cite{CMSGMV18,CFMPT18,TBBB13}.
All approaches to efficiently solve $s$-\textsc{Club} or $s$-\textsc{Plex} in large graphs rely on some form of preprocessing.

\subsection{Our results.}
We transfer the approach of \citet{WB20} to the clique-relaxations $s$-\textsc{Club} and $s$-\textsc{Plex}.
To this end, introducing a new graph parameter, we describe and analyze the Turing kernelization for both problems in \cref{sec:theory}.
Moreover, we provide simple branching algorithms showing fixed-parameter tractability with respect to a gap parameter.

In \cref{sec:exp}, we then analyze the performance of the Turing kernelization in computational experiments for~$s \in \{2,3\}$.
To this end we use ILP-formulations with and without Turing kernelization and basic lower bounds.
For $s$-\textsc{Club} significant speedups are observed whereas for $s$-\textsc{Plex} the improvements are not as clear (though still a speedup factor of more than 2.5 is achieved on average).

In \cref{sec:cor}, we then use correlations (more precisely the Pearson correlation coefficient) to analyze how well our theoretical findings fit to our practically observed running times. 
While this measure makes no statement about the efficiency of the algorithms, we can observe that even with the use of black boxes such as ILP-solvers our theoretical findings are reflected in the experimental results, in particular for the $s$-\textsc{Club} problem.

\section{Preliminaries}\label{sec:prelim}

For an integer~$a \in \mathbb{N}$, we denote by~$[a]$ the set~$\{ 1, \dots, a \}$.
For a graph~$G = (V, E)$, let~$n := |V|$ and~$m := |E|$ be the number of vertices and edges, respectively.
Let~$u, v \in V$ be two vertices of~$G$.
Let~$\dist_G(u, v)$ denote the length of any shortest path between~$u$ and~$v$.
For~$x \in \mathbb{N}$, let~$N_{x, G}(v)$ be the~$x$th neighborhood of~$v$, i.e., the set of vertices~$u$ with~$1 \le \dist_G(u, v) \le x$, $N_{x, G}[v] = \{ v \} \cup N_{x, G}(v)$, and~$\deg_{x, G}(v)$ be the size of its $x$th neighborhood, i.e., $\deg_{x, G}(v) = |N_{x, G}(v)|$. 
For a set~$X \subseteq V$ of vertices, let~$G[X]$ denote the subgraph induced by~$X$.
We drop the subscript~$\cdot_x$ for~$x = 1$. 
Also, we omit the subscript~$\cdot_G$ when~$G$ is clear from context.

\paragraph*{Clique relaxations.}
Let~$X$ be a set of vertices.
If the vertices of~$X$ are pairwise adjacent, then we say that~$X$ is a \emph{clique}.
Let~$s \in \mathbb{N}$ be an integer.
We say that~$X$ is an \emph{$s$-club} if the vertices of~$X$ have pairwise distance at most~$s$, i.e., $\max_{u, v \in X} \dist_{G[X]}(u, v) \le s$ and that~$X$ is an \emph{$s$-plex} if~$G[X]$ is connected and every vertex~$v$ in~$X$ has at most~$s - 1$ vertices nonadjacent to~$v$ in~$X \setminus \{ v \}$, i.e., $\max_{v \in X} |X \setminus N(v)| \le s$.
(Note that a~$1$-club and a~$1$-plex are each a clique.)
We sometimes abuse these terms to refer to the subgraph induced by an $s$-club or $s$-plex.
The decision problems \textsc{$s$-Club} and \textsc{$s$-Plex} ask, given a graph~$G$ and an integer~$k \in \mathbb{N}$, whether~$G$ contains an $s$-club and $s$-plex, respectively, of size at least~$k$. 

\paragraph*{Degeneracy.}
We say that a graph~$G = (V, E)$ is \emph{$d$-degenerate} if for every subgraph~$G'$ of~$G$, there exists a vertex with~$\deg_{G'}(v) \le d$.
Equivalently, $G$ is~$d$-degenerate if there is an ordering of~$V$ in which every vertex has at most~$d$ neighbors that appear later in the ordering.
We say that such an ordering is a \emph{degeneracy ordering} of~$G$.
The degeneracy~$d_G$ of~$G$ is the smallest number~$d$ such that~$G$ is~$d$-degenerate.
For a vertex~$v \in V$ and an ordering~$\sigma$ of~$V$, we denote by~$\Q^{\sigma}_{G}(v)$ (and~$\Q^{\sigma}_G[v]$) the set of vertices in~$N_G(v)$ (and~$N_G[v]$) that appear after~$v$ in $\sigma$.
We also omit the superscript~$\cdot^{\sigma}$ when it is clear.

\paragraph*{Parameterized complexity.}
Here, we list several relevant notions from parameterized complexity.
See e.\,g., \citet{DBLP:books/sp/CyganFKLMPPS15} for a more comprehensive exposition of parameterized complexity.
A parameterized problem is \emph{fixed-parameter tractable} or \emph{FPT} for short if every instance~$(I, k)$ can be solved in time~$f(k) \cdot |I|^{\bigO(1)}$ for some computable function~$f$.
Such an algorithm is called an \emph{FPT algorithm}.
It is widely believed that a parameterized problem is not FPT if it is \emph{W$[i]$-hard} for~$i \in \mathbb{N}$. 
One way to show fixed-parameter tractability is via the notion of \emph{Turing kernel}.
For~$t \in \mathbb{N}$, a \emph{$t$-oracle} for a parameterized problem is an oracle that solves any instance~$(I, k)$ in constant time, provided that~$|I| + k \le t$. 
We say that a parameterized problem admits a Turing kernel of size~$f(k)$ if there is an algorithm with an access to a~$f(k)$-oracle that solves~$(I, k)$ in time~$(|I| + k)^{O(1)}$.
It is straightforward to turn a Turing kernel into an FPT algorithm by simply replacing a~$f(k)$-oracle with a brute-force algorithm.
The brute-force algorithm runs in~$f'(k)$ time for some computable function~$f'$, resulting in a~$f'(k) \cdot (|I| + k)^{\bigO(1)}$-time algorithm.

\section{Theory}\label{sec:theory}

In this section, we provide theoretical analysis of clique relaxations based on the notion of Turing kernels.
We first describe in \Cref{ssec:theory:clique} the algorithm for \textsc{Clique} outlined by \citet{WB20}, which runs in~$1.28^{g} n^{\bigO(1)}$ time for the gap~$g := d - k + 1$.
The algorithms have two components.
The first component is the Turing kernel parameterized by the degeneracy~$d$.
In short, we show that \textsc{Clique} is polynomial-time solvable when we have access to~$f(d)$-oracle (see \Cref{sec:prelim}).
In practice, there is no such convenient oracle so we have to provide some algorithm.
This is the second component.
One way to substitute the oracle is to use a brute-force algorithm.
Since every oracle call takes an input whose size is bounded by~$d$, we already obtain an FPT algorithm parameterized by~$d$.
We can actually make a more refined analysis by considering the gap parameter~$g = d - k + 1$.
Essentially, we use an FPT algorithm parameterized by~$g$ rather than relying on brute force.
\citet{WB20} showed that many \textsc{Clique} instances that can be solved efficiently in practice indeed have small values of~$g$.

We want to adapt this approach to clique relaxations, namely, \textsc{$s$-Club} and \textsc{$s$-Plex}.
However, there is one issue:
Under standard complexity assumptions, there is no FPT algorithm for \textsc{$s$-Club} or \textsc{$s$-Plex}.
More precisely, \textsc{$s$-Club} is known to be NP-hard for~$s = 2$ and~$d = 6$ \cite{HKNS15} and \textsc{$s$-Plex} is known to be W[1]-hard when parameterized by~$d + s$ \cite{DBLP:conf/isaac/KoanaKS20}.
If we were to have a Turing kernel parameterized by the degeneracy~$d$, then it would imply that these problems are FPT with respect to~$d$, which would contradict these results under standard complexity assumptions.
For this reason, we consider a broader notion of degeneracy, which we call \emph{$x$-degeneracy} for~$x \in \mathbb{N}$ (1-degeneracy coincides with the standard degeneracy).
We give the formal definition in \Cref{ssec:theory:degeneracy}.
With the notion of~$x$-degeneracy at hand, we describe how to adapt the approach employed by \citet{WB20} to \textsc{$s$-Club} and \textsc{$s$-Plex}, in \Cref{ssec:theory:club,ssec:theory:plex}, respectively.

\subsection{Algorithm for Clique}
\label{ssec:theory:clique}

\paragraph*{Turing kernel.}
The \textsc{Clique} problem admits a Turing kernel, in which every input to the oracle has at most~$d + 1$ vertices (thus size~$\bigO(d^2)$) as follows:
For an instance~$(G, k)$ of \textsc{Clique}, consider a degeneracy ordering of~$\sigma = (v_1, \dots, v_n)$ of~$G$.
We will assume that~$k \le d + 1$ since a~$d$-degenerate graph has no clique of size~$d + 1$.
Observe that for every clique~$C$ of~$G$, we have~$C \subseteq \Q^\sigma[v]$, where~$v \in C$ is the vertex that appears first in a degeneracy ordering~$\sigma$.
Thus, $G$ has a clique of size~$k$ if and only if there exists a vertex~$v \in V$ such that~$G[\Q[v]]$ has a clique of size~$k$.
Since~$G$ is~$d$-degenerate, $G[\Q[v]]$ has at most~$d + 1$ vertices and size~$\bigO(d^2)$.
This leads to a Turing kernel for the parameter~$d$.

\paragraph*{Oracle algorithm.}
Every oracle can be replaced with a brute-force algorithm running in~$\bigO(2^d d^2)$ time:
Since~$Q[v]$ is of size at most~$d + 1$, there are~$\bigO(2^d)$ subsets of~$Q[v]$ and for every subset, it takes~$\bigO(d^2)$ time to check if every pair of vertices are adjacent.
Thus, \textsc{Clique} can be solved in~$\bigO(2^d \cdot d^2 n)$ time.
In fact, we can refine the analysis for the oracle algorithm in terms of the \emph{gap} parameter~$d - k + 1$:
To that end, we solve the~$\textsc{Deletion to Clique}$ problem:
Given a graph~$G$ and an integer~$\ell$, the task is to find a set of at most~$\ell$ vertices whose deletion results in a clique.
There is a simple~$\bigO(2^{\ell} n^2)$-time algorithm for this problem.
If the vertices are pairwise adjacent and~$\ell \ge 0$, then we have a yes-instance at hand.
Otherwise there exist two nonadjacent vertices, say~$u$ and~$v$.
If~$\ell = 0$, then we can conclude that there is no solution.
If~$\ell \ge 1$, then we recursively solve two instances~$(G - u, \ell - 1)$ and~$(G - v, \ell - 1)$.
This algorithm runs in~$\bigO(2^\ell \cdot n^2)$ time.
Since we need to solve this problem on~$G[Q(v)]$ with~$\ell := |Q[v]| - k \le d - k + 1$,  we have an~$\bigO(2^{d - k + 1} \cdot d^2 n)$-time algorithm for \textsc{Clique}.
We remark that since an instance~$(G, \ell)$ of \textsc{Deletion to Clique} is equivalent to a \textsc{Vertex Cover} instance~$(\overline{G}, \ell)$ for the complement graph~$\overline{G}$ of~$G$, by using a faster known FPT algorithm \cite{XN17}, we obtain an~$\bigO^*_{d}(1.28^{d - k + 1} n)$-time algorithm for \textsc{Clique} ($\bigO^*_d$ hides factors polynomial in~$d$).

\subsection{Extending degeneracy}
\label{ssec:theory:degeneracy}

As mentioned in the beginning of this section, we consider a broader notion of degeneracy defined as follows.

\begin{Definition}
	Let~$G$ be a graph and~$x \in \mathbb{N}$.
	The \emph{$x$-degeneracy} of~$G$ is the smallest integer~$d_x \in \mathbb{N}$ such that for every subgraph~$G'$ of~$G$, there exists a vertex~$v$ with~$|N_{x,G'}(v)| \le d_x$. 
\end{Definition}

The $x$-degeneracy can be formulated in an alternative way:

\begin{Definition}
	Let~$G$ be a graph and~$x \in \mathbb{N}$.
	The \emph{$x$-degeneracy} of~$G$ is the smallest integer~$d_x$ such that there is an ordering~$\sigma = (v_1, \dots, v_n)$ of~$G$ such that for every~$i \in [n]$, the $x$th neighborhood of~$v_i$ in~$G[v_i, \dots, v_n]$ has size at most~$d_x$.
	The ordering~$\sigma$ is called an \emph{$x$-degeneracy ordering}.
	The set of vertices in~$N_{x, G}[v]$ that appear after~$v$ in~$\sigma$ is denoted by~$Q^{\sigma}_{x, G}[v]$.
\end{Definition}

It is not difficult to show that these two definitions are equivalent.
We remark that the notion of 2-degeneracy has been proposed by Trukhanov et al.~\cite{TBBB13} in the context of finding $s$-plexes.
We show that the $x$-degeneracy and an $x$-degeneracy ordering can be found in polynomial time.

\begin{theorem}
	Given a graph~$G$ and an integer~$x \in \mathbb{N}$, we can compute the $x$-degeneracy of~$G$ and an $x$-degeneracy ordering of~$G$ in~$\bigO(n^2 m)$ time.
\end{theorem}
\begin{proof}
	We repeat the following until the graph is empty:
	for every vertex~$v$, we compute the $x$th neighborhood of~$v$.
	We find a vertex whose $x$th neighborhood has the smallest size and delete it from the graph.
	The ordering in which vertices are deleted is an $x$-degeneracy ordering.
	The $x$-degeneracy is the maximum over all vertices of the $x$th neighborhood size when they are deleted.
	Note that we spend~$\bigO(nm)$ time to compute the $x$th neighborhood of every vertex using e.g., BFS.
	Since we repeat this~$n$ times, the algorithm runs in the claimed time.
	\hfill
\end{proof}

\subsection{Algorithm for \textsc{$s$-Club}}
\label{ssec:theory:club}
\paragraph{Turing kernel.}
For \textsc{$s$-Club}, the Turing kernel for \textsc{Clique} can be adapted as follows.
For every~$s$-club~$C$ of~$G$, we have~$C \subseteq Q^{\sigma}[v]$, where~$v \in C$ is the first vertex of~$C$ in an $s$-degeneracy ordering~$\sigma$.
Thus, $G$ has an $s$-club of size~$k$ if and only if there exists a vertex~$v \in V$ such that~$G[\Q[v]]$ has an $s$-club of size~$k$.
By the definition of $x$-degeneracy, we have~$|Q(v)| \le d_s$.
Thus, we have a Turing kernel in which every oracle call involves at most~$d_s + 1$ vertices.

\paragraph{Oracle algorithm.}
Again, we can replace every oracle call with a brute-force algorithm.
The input to every oracle call has at most~$d_s + 1$ vertices and hence there are~$2^{d_s + 1}$ subsets.
Moreover, for every subset, it takes~$\bigO(d_s^3)$ time to determine whether the vertices have pairwise distance at most~$s$, resulting in an algorithm running in~$\bigO(2^{d_s} d_s^3)$ time.
As in \Cref{ssec:theory:clique}, we can also refine the algorithm substituting for the oracle using the parameter~$d_s - k + 1$.
To this end, we solve the \textsc{Deletion to $s$-Club} problem:
Given a graph~$G$ and an integer~$\ell$, the task is to find a set of at most~$\ell$ vertices whose deletion results in an $s$-club.
There is a simple~$\bigO(2^{\ell} n^3)$-time algorithm for this problem.
If~$G$ has diameter at most~$s$ and~$\ell \ge 0$, then we have a yes-instance at hand.
Otherwise there exist two vertices, say~$u$ and~$v$, with~$\dist_G(u, v) > s$.
If~$\ell = 0$, then we can conclude that there is no solution.
If~$\ell \ge 1$, then we recursively solve two instances~$(G - u, \ell - 1)$ and~$(G - v, \ell - 1)$.
Since it takes~$\bigO(n^3)$ time to compute all pairwise distances, this algorithm runs in~$\bigO(2^\ell \cdot n^3)$ time.
Since we need to solve this problem on~$G[Q[v]]$ with~$\ell := |Q[v]| - k \le d_s - k + 1$, we obtain:

\begin{theorem}
	Given the subgraph~$G[Q_s^{\sigma}[v]]$ for every~$v \in V$ for an $s$-degeneracy ordering~$\sigma$, \textsc{$s$-Club} can be solved in~$\bigO(2^{d_s - k} \cdot d_s^3 n)$ time.
\end{theorem}

\subsection{Algorithm for $s$-Plex}
\label{ssec:theory:plex}
\paragraph{Turing kernel.}
For \textsc{$s$-Plex}, we will provide two adaptations.
First, note that every $s$-plex is also an $s$-club and thus the Turing kernel with the parameterization by~$d_s$ follows analogously.
For another adaptation, we use the fact that any $s$-plex with at least~$2s - 1$ vertices have diameter at most two, as observed by \citet{SF78}:
Suppose that two vertices~$u$ and~$v$ in an $s$-plex~$C$ have distance three in~$G[C]$.
Then, every vertex in~$C$ is nonadjacent to either~$u$ or~$v$.
Since for each of~$u$ and~$v$, there are at most~$s - 1$ vertices nonadjacent to it, we have~$|C| \le 2s - 2$.
This leads to a Turing kernel with respect to the parameter~$d_2$ when~$k \ge 2s - 1$.
For every $s$-plex of size at least~$2s - 1$, we have~$C \subseteq Q^\sigma_{2,G}[v]$, where~$v \in C$ is the first vertex of~$C$ in a~$2$-degeneracy ordering~$\sigma$.
Thus, we have again a Turing kernel where every oracle call involves at most~$d_2 + 1$ vertices.

\paragraph{Oracle algorithm.}
Again, we can replace every oracle call with a brute-force algorithm.
The input to every oracle call has at most~$d_2 + 1$ vertices and hence there are~$2^{d_2 + 1}$ subsets.
Moreover, for every subset, it takes~$\bigO(d_2^2)$ time to determine whether it is an $s$-plex, resulting an algorithm running in~$\bigO(2^{d_2} d_2^2)$ time.
As in \Cref{ssec:theory:clique}, we can also refine the algorithm substituting for the oracle using the parameter~$d_2 - k + 1$.
To that end, we solve the \textsc{Deletion to $s$-Plex} problem:
Given a graph~$G$ and an integer~$\ell$, the task is to find a set of at most~$\ell$ vertices whose deletion results in an $s$-plex.
There is a simple~$\bigO((s + 1)^{\ell} n^2)$-time algorithm for this problem.
If~$G$ is an $s$-plex and~$\ell \ge 0$, then we have a yes-instance at hand.
Otherwise there exist a vertex~$v$ and~$s$ vertices nonadjacent to~$v$.
If~$\ell = 0$, then we can conclude that there is no solution.
If~$\ell \ge 1$, then we recursively solve~$s + 1$ instances~$(G - v, \ell - 1)$ and~$(G - u, \ell - 1)$ where~$u$ is one of~$s$ vertices nonadjacent to~$v$.
Since it takes~$\bigO(n^2)$ time to check if the graph is an $s$-plex, this algorithm runs in~$\bigO((s + 1)^\ell \cdot n^2)$ time.
Since we need to solve this problem on~$G[Q[v]]$ with~$\ell := |Q[v]| - k \le d_2 - k + 1$, we obtain:

\begin{theorem}
	Given the subgraph~$G[Q_s^{\sigma}[v]]$ for every~$v \in V$ for an $s$-degeneracy ordering~$\sigma$, \textsc{$s$-Plex} can be solved in time~$\bigO((s + 1)^{d_s - k} \cdot d_s^2 n)$ and~$\bigO(s^2 n^{2s - 1} + (s + 1)^{d_2 - k} \cdot d_2^3 n)$.
\end{theorem}

We remark that for very small~$s$, the first term~$s^2 n^{2s - 1}$ can be ignored in practice, because most instances contain an $s$-plex of size at least~$2s - 1$.

\section{Experiments}\label{sec:exp}

In this section we present the results of our computational experiments for \textsc{$s$-Club} and \textsc{$s$-Plex} for~$s \in \{2,3\}$ on a large dataset of real-world graphs.
We did not optimize every aspect of the implementations as our goal is to investigate the effect of Turing kernelization and the extend to which our theoretical findings are reflected on the running time (this is discussed in \cref{sec:cor}).
We will see, that the Turing kernelization is quite beneficial for \textsc{$s$-Club} but for \textsc{$s$-Plex} the situation is not as clear.

\subsection{Setup}
All experiments were performed on a machine running Ubuntu 18.04 LTS, with an Intel Xeon\textsuperscript{\textregistered{}} W-2125 CPU and 256GB of RAM.
A maximum running time of 1 hour per instance was set.
We used Gurobi 8.1 to solve ILP-formulations, limited to a single thread of execution.
The program that was used to build the ILP models was implemented in C++ and compiled with g++ 7.5.

\paragraph{Dataset.}
The static graphs from the Network Repository \cite{NetRepo} were used for all experiments.
Graphs for which at least one solver configuration timed out, ran out of memory, or completed in less than 0.05 seconds were omitted, in the last case to reduce the effect of noise in the small running time measurements.
The resulting dataset consists of 245 graphs, with 1093 vertices on average.

We remark that \textsc{$s$-Club} and \textsc{$s$-Plex} has been solved for small~$s$ on much larger graphs within minutes~\cite{HKN15,CMSGMV18,CFMPT18}.
The reason we focus on smaller graphs is to have a meaningful multivariate analysis.
More precisely, we want to see if the running time grows (as suggested by theory) with growing $x$-degeneracy and gap.
Having running times for large graphs with small $x$-degeneracy and gap but not for large graphs with large $x$-degeneracy and gap would give misleading results in our analysis in \cref{sec:cor}.

\subsection{Solvers}
We used an ILP solver as oracle for \textsc{$s$-Club} and \textsc{$s$-Plex} in the Turing kernelization.

\paragraph{ILP formulations.}
For \textsc{$s$-Plex} we used a straight-froward formulation with~$\bigO(n)$ variables and constraints and~$\bigO(n+m)$ non-zeroes\footnote{An $s$-plex of size~$\ell > 2s-1$ is guaranteed to be connected and of diameter two~\cite{SF78}. As we only consider~$s\in \{2,3\}$, we do not add constraints enforcing connectedness to the ILP. Unsurprisingly, all found subgraphs were still connected.}.
\begin{align*}
	&\text{maximize:}		& & y \\
	&\text{subject to:}		& & x_v \in \left\{ 0, 1 \right\}, y \in \{0,\ldots,n\}  \\
	&						& & y = \sum_{v \in V} x_v \\
	&\forall v \in V\colon	& & |V|(1-x_v) + \sum_{u \in N(v)} x_u  \geq y - s
\end{align*}
%
%
For \textsc{2-Club} a simplified formulation by \citet{BLP02} was used. 
It has~$\bigO(n)$ variables, $\bigO(n^2)$ constraints, and~$\bigO(n^3)$ non-zeroes.

\begin{align*}
	&\text{maximize:}   						& & \sum_{v \in V} x_v \\
	&\text{subject to:} 						& & x_v \in \left\{ 0, 1 \right\} \\
	&\forall u, v \in V, \dist(u,v) > 2\colon 	& & x_u + x_v \leq 1 \\
	&\forall u, v \in V, \dist(u,v) = 2\colon	& & x_u + x_v \leq 1 + \sum_{c \in N(u) \cap N(v)} x_c   
\end{align*}
%
For \textsc{3-Club} the neighborhood formulation from \citet{AC08,AC14}
was used, with~$\bigO(n)$ binary variables, $\bigO(m)$ continuous variables, at most~$\bigO(n^2)$ and~$\bigO(n^3)$ constraints and non-zeroes, respectively\footnote{For a compact and general ILP formulation for \textsc{$s$-Club} ($s \geq 2$) we refer to \citet{VPP15}.}.
%
\begin{align*}
&\text{maximize:} 									& & \sum_{v \in V} x_v \\
&\text{subject to:} 								& & x_v \in \left\{ 0, 1 \right\} \\
&\dist(u,v) > 3\colon								& & x_u + x_v \leq 1 \\
&\dist(u,v) \in \{2,3\}\colon						& & x_u + x_v \leq 1 + \smashoperator{\sum_{c \in N(u) \cap N(v)}} x_c + \sum_{e \in E_{uv}} z_e  \\
&\forall e=\{a,b\} \in E\colon						& & z_e \leq x_a, \,\, z_e \leq x_b \\
&\forall e \in E\colon								& & 0 \leq z_e \leq 1
\end{align*}
where~$E_{uv} = \{ \{p,q\} \in E \mid p \in N(u) \setminus N(v), q \in N(v) \setminus N(u) \}$.

\paragraph{Solver variants.}
We tested several different approaches using these ILP models, each reflecting one stage of the concepts in \cref{sec:theory}.
To this end, we use four different solver configurations, namely \texttt{noTK, hint, default} and \texttt{full} (described below).
We will refer to, for example, \texttt{2club\_noTK} as the benchmark results of the \texttt{noTK} solver configuration on the \textsc{2-Club} problem.

One solver variant simply built a single ILP model for the entire graph, which we call the \texttt{noTK} variant (no Turing kernel).
All other variants use the Turing kernelization to some extent.
The \texttt{full} variant makes only basic use of Turing kernelization, utilizing the 2/3-degeneracy as described in \cref{sec:theory}.
There, each oracle call is solved via an ILP.
The solution size is then the maximum solution size over all cores.

For \textsc{2-Club} and \textsc{2-Plex} the Turing kernelization using 2-degeneracy is employed, for \textsc{3-Club} and \textsc{3-Plex} the one using 3-degeneracy.
As there is only one connected 3-plex of diameter three (the~$P_4$) which was never the largest 3-plex in our experiments we also used the 2-degeneracy based Turing kernelization for \textsc{3-Plex}.
We report the results of this variant under \texttt{3plex-2}.
%

The \texttt{default} variant uses the Turing kernel approach in combination with a simple lower bound:
It uses the maximum solution size of already solved ILPs as a lower bound on the global solution size by adding a constraint to the ILPs enforcing that the solution has to be larger than the current lower bound. 
The order in which the ILPs are solved can therefore have an impact on the overall running time.
We did not analyze this effect and used a fixed 2/3-degeneracy ordering.
Instead, we remove this effect in the \texttt{hint} variant.
%
There, we added a constraint to each ILP model in the Turing kernels which enforced that the solution size to the Turing kernel is at least the size of the global solution size for the entire graph.
Thus, one can think of (heuristically) optimizing in the \texttt{default} variant the order in the Turing kernelization so that the oracle calls giving the largest results come first.
Alternatively, this shows the maximum speedup possible by a ``perfect'' heuristic.
Note that in the \texttt{hint} variant at least one ILP model still has a feasible solution.


\subsection{Results}

\newcommand{\compareMarkers}{
	\addplot[color=black,domain=\minValueRun:\maxValueRun,samples=4] {x};
	\addplot[dashed,color=black!75,domain=\minValueRun:\maxValueRun,samples=4] {5*x};
	\addplot[dashed,color=black!75,domain=\minValueRun:\maxValueRun,samples=4] {0.2*x};
	\addplot[dotted,color=black,domain=\minValueRun:\maxValueRun,samples=4] {25*x};
	\addplot[dotted,color=black,domain=\minValueRun:\maxValueRun,samples=4] {0.04*x};

}
\newcommand{\compareGroupPlotGeneral}[6]{
	\nextgroupplot[#4, ylabel={#1}]
		\addplot+ table[col sep=comma,y={runtime_default}, x={runtime_notk}] {#2};
		\compareMarkers{}
		\addplot[color=red,mark=none] coordinates {(\minValueRun, 3600) (\maxValueRun, 3600)};
		\addplot[color=red,mark=none] coordinates {(3600, \minValueRun) (3600, \maxValueRun)};

	\nextgroupplot[#5]
		\addplot+ table[col sep=comma,y={runtime_default}, x={runtime_notk}] {#3};
		\compareMarkers{}
		\addplot[color=red,mark=none] coordinates {(\minValueRun, 3600) (\maxValueRun, 3600)};
		\addplot[color=red,mark=none] coordinates {(3600, \minValueRun) (3600, \maxValueRun)};

	\nextgroupplot[#4, #6]
		\addplot+ table[col sep=comma,y={runtime_default}, x={runtime_hint}] {#2};
		\compareMarkers{}
		\addplot[color=red,mark=none] coordinates {(\minValueRun, 3600) (\maxValueRun, 3600)};
		\addplot[color=red,mark=none] coordinates {(3600, \minValueRun) (3600, \maxValueRun)};
}%
\newcommand{\compareGroupPlotTop}[3]{\compareGroupPlotGeneral{#1}{#2}{#3}{title={$s = 2$}}{title={$s = 3$}}{}}%
\newcommand{\compareGroupPlot}[3]{\compareGroupPlotGeneral{#1}{#2}{#3}{}{}{xlabel={\texttt{hint} [s]}}}%
\begin{figure}[t!]
	\centering
	\def\maxValueRun{10000}
	\def\minValueRun{0.001}
	
	\ifexternalize\tikzsetnextfilename{compare-running-times}\fi
	\begin{tikzpicture}
		\begin{groupplot}[
				group style={
					group name=my plots, 
					group size=3 by 2, 
					xlabels at=edge bottom, 
					ylabels at=edge left, 
					xticklabels at=edge bottom, 
					yticklabels at=edge left, 
					vertical sep=4pt,
					horizontal sep=4pt
				}, 
				ymode=log,
				grid,
				xmode=log,
				xmin=0.03,
				xmax=6000,
				ymin=0.03,
				ymax=6000,
				ytick={0.01,0.1,1,10,100,1000},
				xtick={0.01,0.1,1,10,100,1000},
				width=0.38\hsize, 
				height=0.3\hsize, 
				xlabel={\texttt{noTK} [s]},
				cycle multiindex* list = {only marks \nextlist mark=x}, 
			] 
			\compareGroupPlotTop{\texttt{default} (clubs)}{data/2club-comparison.csv}{data/3club-comparison.csv}

			\compareGroupPlot{\texttt{default} (plexes)}{data/2plex-comparison.csv}{data/3plex-comparison.csv}
		\end{groupplot}
	\end{tikzpicture} 
	
	
	\caption{
		Running time comparison (in seconds) of different variants (top row for~$2$- and $3$-\textsc{Club}, second row for~$2$- and $3$-\textsc{Plex}).
		Each cross represents one instance with the~$x$- and~$y$-coordinates indicating the running time of the respective variant (in seconds): \texttt{default} and \texttt{noTK}.
		Thus, a cross above (below) the solid diagonal indicates that the solver on the $x$-axis ($y$-axis) is faster on the corresponding instance.
		The diagonal lines mark factors of~$1$ (solid), $5$ (dashed) and~$25$ (dotted). 
		The solid horizontal red lines (at 3600 seconds) indicate the time limit.
		For~$2$- and $3$-\textsc{Club} a significant running time improvement is visible.
		For~$2$- and $3$-\textsc{Plex} the picture is not so clear.
	}
	\label{fig:compare-running-times}
\end{figure}
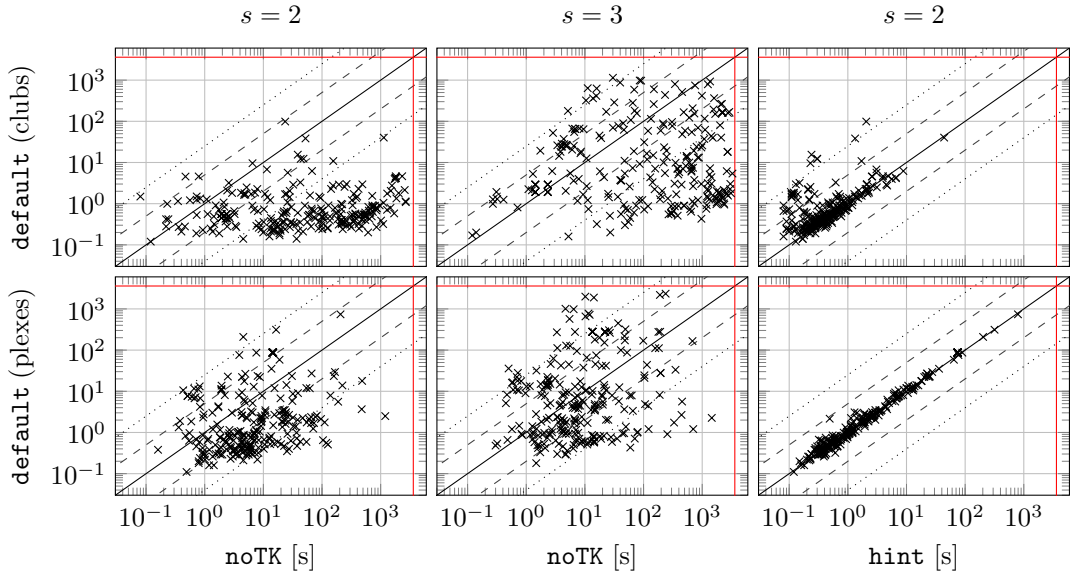

In \cref{tab:avg_runtime} we summarize the average running time of the different solver configurations on the four considered problems.
\begin{table}
	\caption{Average running times in seconds of various solver configurations.}
	\pgfplotstableread[col sep = comma]{data/avg_runtime.csv}\data 
	\centering
	\pgfplotstabletypeset[columns={problem,notk,full,default,hint},
		columns/problem/.style={string type,column name=,column type = {l}},
		columns/notk/.style={precision=1,fixed zerofill,fixed,column name=noTK,column type = {r}},
		columns/hint/.style={precision=1,fixed zerofill,fixed,column name=hint,column type = {r}},
		columns/default/.style={precision=1,fixed zerofill,fixed,column name=default,column type = {r}},
		columns/full/.style={precision=1,fixed zerofill,fixed,column name=full,column type = {r}},
		every head row/.style ={before row=\toprule, after row=\midrule},
	    every last row/.style ={after row=\bottomrule}]{\data}
	\label{tab:avg_runtime}
\end{table}
The approach without Turing kernels is significantly slower for~$2$- and $3$-\textsc{Club} but not so much for~$2$- and $3$-\textsc{Plex}.
This can also be seen in the detailed comparisons in \cref{fig:compare-running-times}.
%
Interestingly, the \texttt{noTK} variants are much faster in finding plexes than in finding clubs (more than 10 times larger average running time).
However, the average running time of \texttt{2club\_default} is five times smaller than that of \texttt{2plex\_default}.
Thus, on the one hand the ILP-formulation we use for finding 2/3-clubs may have a room for improvement.
On the other hand, the Turing kernel approach works much better for clubs than plexes.
The reason is probably that the Turing kernels are built based on distance, which fits better with clubs than plexes.

Unsurprisingly, the \texttt{hint} variant is the fastest one.
However, the \texttt{default} variant is nearly as fast as \texttt{hint} (see also right column of \cref{fig:compare-running-times}), even though it does not receive the solution size as input, and instead uses the maximum solution of the previously solved ILPs to update the lower bound.
Moreover, the \texttt{default} variant is considerably faster than the \texttt{full} variant.
This shows the strength of the lower bounds employed in the ILP-solver.
Remarkably, for finding 2/3-clubs even the \texttt{full} variant brings a decent speedup compared to the \texttt{noTK} variant.


\section{Correlation between Theory and Practice} \label{sec:cor}

\begin{figure}[t!]
	\centering
	\def\maxX{7600}
	\def\minX{-300}
	\def\maxXX{1250}
	\def\minXX{-50}
	\def\maxXXX{170}
	\def\minXXX{-7}

	\ifexternalize\tikzsetnextfilename{running-times-versus}\fi
	\begin{tikzpicture}
		\begin{groupplot}[
				group style={
					group name=my plots, 
					group size=3 by 4, 
					xlabels at=edge bottom, 
					ylabels at=edge left, 
					xticklabels at=edge bottom, 
					yticklabels at=edge left, 
					vertical sep=4pt,
					horizontal sep=4pt
				}, 
				ymode=log,
				grid,
				xmin=\minX,
				xmax=\maxX,
				ymax=200,
				width=0.4\hsize, 
				height=0.33\hsize, 
				%
				cycle multiindex* list = {only marks \nextlist mark=x}, 
				legend style={ at={(0.5,0.95)}, anchor=north,}
			] 
			\nextgroupplot[ylabel={time [s] -- \texttt{noTK}},legend style={ at={(0.95,0.05)}, anchor=south east,},ymax=6000]
				\addplot+[] table[col sep=comma,x={n}, y={runtime_notk}] {data/2club-comparison.csv};
				\addplot[red,domain=\minX:\maxX,samples=3] {11.15998 * 1.00134^x};

			\nextgroupplot[xmin=\minXX,xmax=\maxXX,legend style={ at={(0.95,0.05)}, anchor=south east,},ymax=6000]
				\addplot+[] table[col sep=comma,x={gdegen}, y={runtime_notk}] {data/2club-comparison.csv};
				\addplot[red,domain=\minXX:\maxXX,samples=3] {53.59703 * 0.99909^x};
				
			\nextgroupplot[xmin=\minXXX,xmax=\maxXXX,legend style={ at={(0.95,0.05)}, anchor=south east,},ymax=6000]
				\addplot+[] table[col sep=comma,x expr=(\thisrow{gdegen} - \thisrow{solution}), y={runtime_notk}] {data/2club-comparison.csv};
				\addplot[red,domain=\minXXX:\maxXXX,samples=3] {46.95778 * 1.00086^x};

			\nextgroupplot[ylabel={time [s] -- \texttt{full}},legend style={ at={(0.95,0.05)}, anchor=south east,},ymax=6000]
				\addplot+[] table[col sep=comma,x={n}, y={runtime_full}] {data/2club-comparison.csv};
				\addplot[red,domain=\minX:\maxX,samples=3] {1.74559 * 1.00055^x};

			\nextgroupplot[xmin=\minXX,xmax=\maxXX,legend style={ at={(0.95,0.05)}, anchor=south east,},ymax=6000]
				\addplot+[] table[col sep=comma,x={gdegen}, y={runtime_full}] {data/2club-comparison.csv};
				\addplot[red,domain=\minXX:\maxXX,samples=3] {1.21246 * 1.00701^x};
				
			\nextgroupplot[xmin=\minXXX,xmax=\maxXXX,legend style={ at={(0.95,0.05)}, anchor=south east,},ymax=6000]
				\addplot+[] table[col sep=comma,x expr=(\thisrow{gdegen} - \thisrow{solution}), y={runtime_full}] {data/2club-comparison.csv};
				\addplot[red,domain=\minXXX:\maxXXX,samples=3] {2.98625 * 1.00667^x};

			\nextgroupplot[ylabel={time [s] -- \texttt{default}}]
				\addplot+[] table[col sep=comma,x={n}, y={runtime_default}] {data/2club-comparison.csv};
				\addplot[red,domain=\minX:\maxX,samples=3] {0.69610 * 1.00003^x};

			\nextgroupplot[xmin=\minXX,xmax=\maxXX]
				\addplot+[] table[col sep=comma,x={gdegen}, y={runtime_default}] {data/2club-comparison.csv};
				\addplot[red,domain=\minXX:\maxXX,samples=3] {0.54771 * 1.00198^x};

			\nextgroupplot[xmin=\minXXX,xmax=\maxXXX]
				\addplot+[] table[col sep=comma,x expr=(\thisrow{gdegen} - \thisrow{solution}), y={runtime_default}] {data/2club-comparison.csv};
				\addplot[red,domain=\minXXX:\maxXXX,samples=3] {0.57924 * 1.02599^x};

			\nextgroupplot[ylabel={time [s] -- \texttt{hint}},xlabel={$n$}]
				\addplot+[] table[col sep=comma,x={n}, y={runtime_hint}] {data/2club-comparison.csv};
				\addplot[red,domain=\minX:\maxX,samples=3] {0.47079 * 0.99994^x};

			\nextgroupplot[xmin=\minXX,xmax=\maxXX,xlabel={2-degeneracy}]
				\addplot+[] table[col sep=comma,x={gdegen}, y={runtime_hint}] {data/2club-comparison.csv};
				\addplot[red,domain=\minXX:\maxXX,samples=3] {0.44332 * 0.99996^x};
				
			\nextgroupplot[xmin=\minXXX,xmax=\maxXXX,xlabel={gap}]
				\addplot+[] table[col sep=comma,x expr=(\thisrow{gdegen} - \thisrow{solution}), y={runtime_hint}] {data/2club-comparison.csv};
				\addplot[red,domain=\minXXX:\maxXXX,samples=3] {0.30724 * 1.04403^x};
				
%
%
		\end{groupplot}
	\end{tikzpicture} 
	
	
	\caption{
		The running times of \texttt{2club\_noTK}, \texttt{2club\_full}, \texttt{2club\_default}, and \texttt{2club\_hint} plotted against the number~$n$ of vertices, the 2-degeneracy, and the gap of the input graph.
		The solid red lines are linear regressions best fitting to the data points (where the logarithm of the running times is taken).
	}
	\label{fig:running-times-versus}
\end{figure}
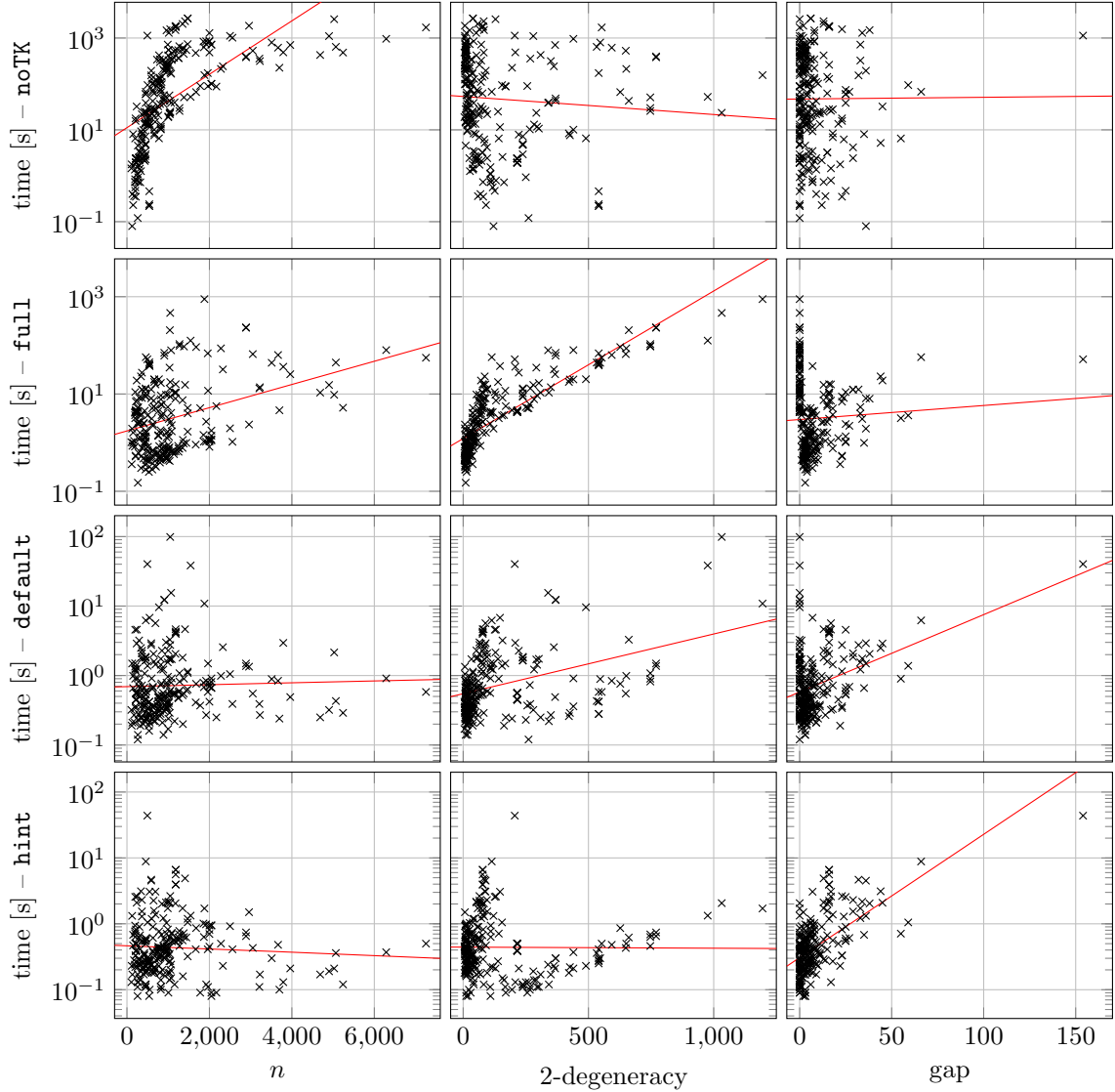

Given the theoretical running time bounds in \cref{sec:theory} and the measured running times in \cref{sec:exp}, we now analyze correlations between these.
Since we have an NP-hard problem, our working hypothesis is that the running time should depend exponentially on some parameter(s).
Natural parameter candidates are the number~$n$ of vertices, the 2- resp.\ 3-degeneracy, and the gap parameter.
The latter parameters are suggested by our theoretical findings.
We studied five problems (counting \texttt{3plex} and \texttt{3plex-2} as two) with four different solvers of each problem.
Thus, there are~$5 \cdot 4 \cdot 3 = 60$ different parameter -- running time pairs to analyze.
We depicted the 12 pairs for \texttt{2club} in \cref{fig:running-times-versus} with the red lines depicting the exponential function of the form~$\alpha^p \cdot \beta$ that best fit the data; these lines are computed via linear regression (logarithm of the running time versus parameter value~$p$).
Obviously, the suggested running time function on the bottom left (\texttt{2club\_hint} with parameter~$n$) is useless: our implementation will in general not become faster the larger the input gets.
The plots on the top left and bottom right seem more sensible.

\subsection{Method}
Instead of ``carefully looking'' at each of the plots in \cref{fig:running-times-versus} and finding arguments for each one of them, we want an automated way of distinguishing sensible from useless suggestions.
To this end, we suggest using the Pearson correlation coefficient which we subsequently just call correlation coefficient.
It is a standard measure of linear correlation between two sets of data.
Simply put, given the (running time, parameter) data points the correlation coefficient computes a number between -1 and 1. 
If there is no correlation at all, then the coefficient is 0.
With perfect correlation (i.\,e.\ the data points are on a straight line with positive slope) the coefficient is 1.
With perfect negative correlation (i.\,e.\ the data points are on a straight line with negative slope) it is -1.
Hence, in our case the numbers close to one describe a good (linear!) correlation between the parameter and the (logarithm of the) running time.

Before we present the correlation coefficients for our experiments and parameters, let us give some disclaimers for our particular setting. 
\begin{itemize}
	\item As we are (for now) only interested in simple exponential dependencies, the Pearson correlation coefficient suffices as we can take the logarithm of all measured running times. 
	There are different correlation coefficients that can also measure non-linear correlations and might be better suited to other settings.
	
	\item A better correlation coefficient does not imply a better running time, just a better correlation with the respective parameter. 
	
	\item We use a very simplistic analysis. 
	For example, we do not discuss confidence intervals or similar issues. 
	The reason being that any ``good'' correlation between a parameter (or a combination of parameters) and the measured running time is only an \emph{indication} for such a correlation. 
	In particular, if some ``new'' correlations are discovered with this method, then this only gives suggestions. 
	We still have to mathematically prove the running times.
	Moreover, if theoretical hardness results (e.\,g.\ NP-hardness for a constant parameter value) disproves the correlation, then new explanations have to be found (e.\,g.\ there are several other parameters that, in combination, also yield the correlation and allow for provable running time bounds).
	
	\item We restrict ourselves to correlations between one parameter and the running time.
		While correlations between multiple parameters and the running time are possible, our theoretical results in \cref{sec:theory} only suggest exponential dependencies between one parameter (2/3-degeneracy or gap) and the running time and not two parameters.
		Incorporating the polynomial factors in the running times of \cref{sec:theory} is possible, but in our analysis it changed the coefficients only marginally (by less than 5\%, usually much less than 1\%).
\end{itemize}

\subsection{Results}


\cref{tab:correlations} summarizes the 60 correlation coefficients of three graph parameters with the logarithm of the measured running times.

\begin{table} 
	\caption{
	Tables summarizing the correlation of different graph parameters~$n$ (left table), $d_2/d_3$ (middle table), and the gap~$g$ (right table) with the logarithm of the measured running times of various solver configurations (def abbreviates default).
	In the middle the correlation with 2-degeneracy is shown for \texttt{2club, 2plex} and \texttt{3plex-2}, and with 3-degeneracy for \texttt{3club} and \texttt{3plex}).
	}
	\label{tab:correlations}
	\setlength{\tabcolsep}{4.5pt}
	
		\pgfplotstableread[col sep = comma]{data/n_vs_logruntime_correl.csv}\data 
		\centering
		\pgfplotstabletypeset[columns={problem,notk,full,default,hint},
			columns/problem/.style={string type,column name=,column type = {l}},
			columns/notk/.style={precision=2,fixed,column name=noTK,column type = {r}},
			columns/hint/.style={precision=2,fixed,column name=hint,column type = {r}},
			columns/default/.style={precision=2,fixed,column name=def,column type = {r}},
			columns/full/.style={precision=2,fixed,column name=full,column type = {r}},
			every head row/.style ={before row=\toprule, after row=\midrule},
			every last row/.style ={after row=\bottomrule}]{\data}
	\hfill
		\pgfplotstableread[col sep = comma]{data/gdegen_vs_logruntime_correl.csv}\data 
		\centering
		\pgfplotstabletypeset[columns={notk,full,default,hint},
			columns/notk/.style={precision=2,fixed zerofill,fixed,column name=noTK,column type = {r}},
			columns/hint/.style={precision=2,fixed zerofill,fixed,column name=hint,column type = {r}},
			columns/default/.style={precision=2,fixed zerofill,fixed,column name=def,column type = {r}},
			columns/full/.style={precision=2,fixed zerofill,fixed,column name=full,column type = {r}},
			every head row/.style ={before row=\toprule, after row=\midrule},
			every last row/.style ={after row=\bottomrule}]{\data}
	\hfill
		\pgfplotstableread[col sep = comma]{data/gap_vs_logruntime_correl.csv}\data 
		\centering
		\pgfplotstabletypeset[columns={notk,full,default,hint},
			columns/notk/.style={precision=2,fixed zerofill,fixed,column name=noTK,column type = {r}},
			columns/hint/.style={precision=2,fixed zerofill,fixed,column name=hint,column type = {r}},
			columns/default/.style={precision=2,fixed zerofill,fixed,column name=def,column type = {r}},
			columns/full/.style={precision=2,fixed zerofill,fixed,column name=full,column type = {r}},
			every head row/.style ={before row=\toprule, after row=\midrule},
			every last row/.style ={after row=\bottomrule}]{\data}
\end{table}
%

Consider the first row corresponding to \textsc{2-Club} in \cref{tab:correlations}.
The first columns display the correlation with~$n$ which is best for the \texttt{noTK} variant.
This well reflects our observations for the plots in the left column of \cref{fig:running-times-versus}: 
The \texttt{default} and \texttt{hint} variant do not display any reasonable correlation with~$n$, only \texttt{noTK} does to some extend.
Similarly, in the right column of \cref{fig:running-times-versus} the correlations of the \texttt{default} and \texttt{hint} variant with the gap-parameter are quite decent, but not for the the \texttt{noTK} variant.
Moreover, in the middle plot (\texttt{default} variant) of the right column in \cref{fig:running-times-versus} there are a few instances that have a high running time despite a parameter value of zero. 
This is an argument against the suggested regression being a ``good'' explanation.
Also, in the bottom right plot one can see that there are no such (drastic) outliers.
Hence, the correlation with the \texttt{hint} variant with the gap is considerably ``better'' than with the slower \texttt{default} variant.
This is also reflected in the corresponding correlation coefficients of 0.61 and 0.35 respectively (see two rightmost columns in \cref{tab:correlations}) and, thus, supports the correlation coefficient as reasonable measure.

The results for the other problems are somewhat similar to the ones for \textsc{2-Club}.
The correlation coefficient for the number of vertices is highest for all problems with the \texttt{noTK} configuration, whereas with the configurations based on Turing kernels it is significantly lower (or even negative). 
This is somehow expected, as all our ILP formulations use~$\bigO(n)$ binary variables.
State-of-the-art ILP solvers are highly complex (``a bag of tricks'') and able to solve instances with millions of integer variables efficiently.
Thus, the correlation of around 0.5 (for \texttt{noTK}) with the number of vertices is higher than for the other variants, but not the overall highest correlations (see second row and second column in \cref{fig:running-times-versus} for the plot corresponding to the highest correlation).


The Turing kernel approaches correlate in general better with the 2/3-degeneracy, notable exceptions are the \texttt{2club\_hint} and \texttt{3club\_hint} variants.
As expected, across all problems the highest correlations with the 2/3-degeneracy are achieved by the \texttt{full} variants:
The 2-degeneracy (3-degeneracy) is in our dataset on average more than five times (more than three times) smaller than~$n$.
Hence, the high correlations for the \texttt{noTK} variants with~$n$ translate to high correlations for the \texttt{full} variants with 2/3-degeneracy.
For the \texttt{noTK} configuration there is barely any correlation with the 2/3-degeneracy.
It thus seems that the ILP solver cannot exploit the 2/3-degeneracy---at least with the given ILP formulations.

As discussed in \cref{sec:exp}, the \texttt{default} and \texttt{hint} variants are considerably faster than the \texttt{full} variants due to having access to some (perfect) lower bound.
As we use the black box of an ILP-solver we do not have theoretical running time bounds covering the effect of this lower bound. 
Nevertheless, the correlation coefficients support some speculations:
The correlations in the middle table of \cref{tab:correlations} suggest that this running time improvement is not (so much) correlated to the 2/3-degeneracy but to another parameter.
For finding clubs the gap-parameter is a good explanation: \texttt{2club\_hint} and \texttt{3club\_hint} have high correlations with the gap parameter. 
Thus, with a better lower bound computation (i.\,e., some actual heuristic) we suspect the correlation of the \texttt{default} variant with the degeneracy to decrease and increase with the gap parameter.

For plexes this argumentation does not hold.
There seems rarely any difference in the correlation coefficients of the \texttt{default} and \texttt{hint} variant with the 2/3-degeneracy and the gap-parameter.
The reason is simple: while the gap is considerably smaller than the 2/3-degeneracy for 2/3-\textsc{Club}, this is not the case for 2/3-\textsc{Plex}, see \cref{fig:degeneracy-vs-gap}.
\begin{figure}[t]
	\centering
	\def\maxValueRun{5000}
	\def\minValueRun{0.2}
	\ifexternalize\tikzsetnextfilename{degeneracy-vs-gap}\fi
	\begin{tikzpicture}
		\begin{groupplot}[
				group style={
					group name=my plots, 
					group size=2 by 2, 
					xlabels at=edge bottom, 
					ylabels at=edge left, 
					xticklabels at=edge bottom, 
					yticklabels at=edge left, 
					vertical sep=4pt,
					horizontal sep=4pt
				}, 
				ymode=log,
				grid,
				xmode=log,
				xmin=\minValueRun,
				xmax=\maxValueRun,
				ymin=\minValueRun,
				ymax=\maxValueRun,
				ytick={0.01,0.1,1,10,100,1000},
				xtick={0.01,0.1,1,10,100,1000},
				width=0.5\hsize, 
				height=0.4\hsize, 
				cycle multiindex* list = {only marks \nextlist mark=x}, 
			]  
			\nextgroupplot[ylabel={\texttt{gap} (clubs)}]
				\addplot[only marks, mark=x] table[col sep=comma,x={gdegen}, y expr=(\thisrow{gdegen} - \thisrow{solution} +0.5)] {data/2club-comparison.csv};
				\compareMarkers{}

			\nextgroupplot[]
				\addplot[only marks, mark=x] table[col sep=comma,x={gdegen}, y expr=(\thisrow{gdegen} - \thisrow{solution} +0.5)] {data/3club-comparison.csv};
				\compareMarkers{}
				
			\nextgroupplot[ylabel={\texttt{gap} (plexes)},xlabel={\texttt{2-degeneracy}}]
				\addplot[only marks, mark=x] table[col sep=comma,x={gdegen}, y expr=(\thisrow{gdegen} - \thisrow{solution} +0.5)] {data/2plex-comparison.csv};
				\compareMarkers{}

			\nextgroupplot[xlabel={\texttt{3-degeneracy}}]
				\addplot[only marks, mark=x] table[col sep=comma,x={gdegen}, y expr=(\thisrow{gdegen} - \thisrow{solution} +0.5)] {data/3plex-comparison.csv};
				\compareMarkers{}
		\end{groupplot}
	\end{tikzpicture} 
	\caption{
		Relation between 2-degeneracy and gap.
	}
	\label{fig:degeneracy-vs-gap}
\end{figure}
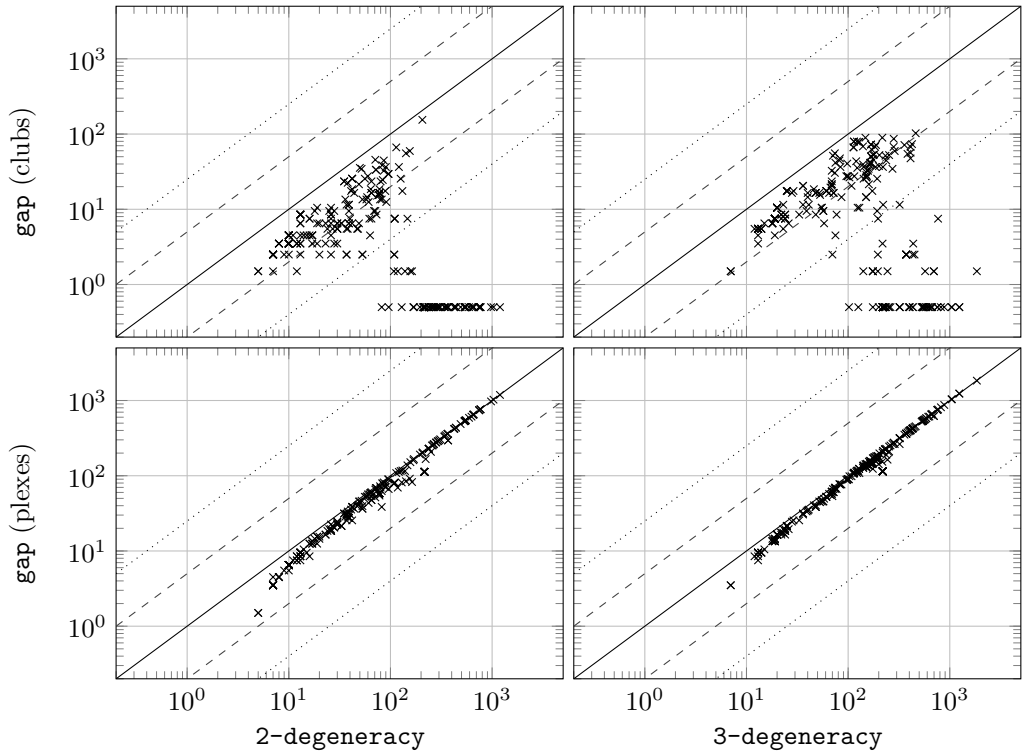
Thus, for 2/3-Plex the correlations differ only marginally between the \texttt{hint} and \texttt{default} variants.
Moreover, this explains very well why despite \texttt{2club\_noTK} being quite slow compared to \texttt{2plex\_noTK} the variant \texttt{2club\_hint} is much faster than \texttt{2plex\_hint}: The average gap for our 2-club instances is~$7.9$, hence the exponential running time dependency on the gap is manageable.
For 3-club instances, the average gap is 18.6 which, apparently, is one of the reasons why the \texttt{3club} variants are much slower than the \texttt{2club} variants.

\section{Conclusion}\label{sec:conclusion}

We provided theoretical bounds for algorithms solving $s$-\textsc{Club} and $s$-\textsc{Plex} and experimentally tested the employed Turing kernelization for~$s \in \{2,3\}$.
More importantly, we discussed the correlation between the observed running times and the theoretical bounds.
Yet, there is still a large gap between theory and practice: for example, the bases of the exponential function obtained by regression are all below 1.1---much smaller than current theoretical results suggest.
We are confident that the use of correlation coefficients as demonstrated in our work can help to close this gap. 
They are easy to employ and quite flexible.
We see the following directions for future work:
\begin{itemize}
	\item We found that the Turing kernel approach improves the runtime significantly more for clubs than plexes. We believe that this is due to the fact that the $x$-degeneracy is defined based on distance. Is there an analogous notion more suitable for finding plexes?
	\item Checking whether the running times correlates with multiple parameters is an easy extension.
	The whole process should allow for relatively easy automation. 
	An automated tool could generate a list of likely correlations from experimental results.
	These can then be analyzed theoretically with the parameterized complexity framework.
	This way, practice could give more impulses for theory.
	\item
	In this work, we only investigated the correlation coefficients between the runtime and graph parameters.
	Although this approach has clear advantages in simplicity, there are some drawbacks; for instance, the correlation between parameters are completely overlooked.
	One could perhaps sharpen the analysis using a more sophisticated statistical method.
	\item The approach is not limited to analyzing running times. 
	Other objectives could be the size of preprocessed instances (using the kernelization framework from parameterized algorithmics) or approximation factors of heuristics or approximation algorithms.
	\item While we use worst-case analysis, average case analysis or smoothed analysis are also suitable for the approach.
\end{itemize}
Of course there are downsides to the approach.
For example, how to incorporate timeouts? 
Or are different correlation coefficients better suited?
Addressing these issues is another task for future work.

\printbibliography[heading=bibintoc,title={References}]

\section*{Acknowledgement}
Tomohiro Koana is supported by the Deutsche Forschungsgemeinschaft (DFG) project DiPa (NI 369/21).

\end{document}
